\documentclass[preprint,12pt]{elsarticle}
\usepackage [utf8]{inputenc}
\usepackage{amsmath}
\usepackage{amsthm}
\usepackage{graphicx}
\usepackage{pdfpages}
\usepackage{lineno,hyperref}
\usepackage{amssymb}
\hypersetup{
colorlinks=true,
linkcolor=cyan,
filecolor=magenta, urlcolor=blue,
}
\newtheorem{theorem}{Theorem} [section]

\begin{document}


\begin{frontmatter}

\title{Optimal Lockdown Policy for Covid-19: \\A Modelling Study}

\author{Yuting Fu \qquad Haitao Xiang \qquad Hanqing Jin \qquad Ning Wang}
\address{Mathematical Institute, University of Oxford, Oxford, UK}

\begin{abstract}

As the COVID-19 spreads across the world, prevention measures are becoming the essential weapons to combat against the pandemic in the period of crisis. The lockdown measure is the most controversial one as it imposes an overwhelming impact on our economy and society. Especially when and how to enforce the lockdown measures are the most challenging questions considering both economic and epidemiological costs. In this paper, we extend the classic SIR model to find optimal decision making to balance between economy and people's health during the outbreak of COVID-19. In our model, we intend to solve a two phases optimisation problem: policymakers control the lockdown rate to maximise the overall welfare of the society; people in different health statuses take different decisions on their working hours and consumption to maximise their utility. We develop a novel method to estimate parameters for the model through various additional sources of data. We use the Cournot equilibrium to model people's behaviour and also consider the cost of death in order to leverage between economic and epidemic costs. The analysis of simulation results provides scientific suggestions for policymakers to make critical decisions on when to start the lockdown and how strong it should be during the whole period of the outbreak. Although the model is originally proposed for the COVID-19 pandemic,  it can be generalised to address similar problems to control the outbreak of other infectious diseases with the lockdown measures. 

\end{abstract}


\end{frontmatter}





\section{Introduction}  

As one of the most devastating pandemics in human history, the current outbreak of COVID-19 has already caused more than one million deaths around the globe. Numerous prevention measures have been studied by \cite{flaxman2020estimating} and \cite{tian2020investigation} in order to control the spread of the virus by the governments. For example, medical measures, such as research on testing, medicine, and vaccine, are accelerated; relatively easy measures, like face masking and social distancing, are also widely accepted and applied. Essentially the most effective prevention of COVID-19 is the lockdown measure which completely ceases the movement of the human being and thus slows down the spread of disease. However, the lockdown measure is incredibly controversial as it imposes a tremendous impact on our society and economy. Hence it might be the most difficult decision to be made by the governments. Especially when and how to impose the lockdown measure is one of the most challenging questions for both politicians and scientists. To address this question, there is a need to develop a mathematical model combining both epidemiology and economics. 

Epidemiological models have been widely studied to analyse the dynamics of the pandemic  (\cite{kucharski2020early} \cite{liu2020reproductive} \cite{wang2020evaluation} \cite{wu2020nowcasting}). However, there is less discussion on how the lockdown policies can influence the economic decisions of people and the spread of disease and how can policymakers make optimal policy in the epidemic. \cite{ferguson2020impact} and \cite{ferguson2005strategies} analyse the government intervention using epidemiological models with exogenous parameters and evaluate the effect of the intervention by simulation results. Some recent papers focus on analysis of optimal policy and policy effect in the framework of the SIR model or its variants. They studied the effect of different measures including fiscal policy  (\cite{eichenbaum2020macroeconomics},  \cite{faria2020fiscal}, \cite{guerrieri2020macroeconomic}), testing and quarantine (\cite{piguillem2020optimal}, \cite{berger2020seir},
 \cite{jenny2020dynamic}, \cite{aleta2020modelling}), intervention policy on multi-aged groups (\cite{brotherhood2020economic}, \cite{acemoglu2020optimal}, \cite{gollier2020cost}), social distancing  ( \cite{jones2020optimal}, \cite{farboodi2020internal}, \cite{gourinchas2020flattening}) and lockdown control (\cite{alvarez2020simple},  \cite{gonzalez2020optimal},  \cite{acemoglu2020optimal}).  In previous works,  \cite{alvarez2020simple} studied the optimal lockdown policy that minimises the value of fatalities and the output costs of the lockdown policy by locking down part of the susceptible and Infectious population,  \cite{acemoglu2020optimal} researched the optimal lockdown policies on people of different age groups, and  \cite{gonzalez2020optimal} maximise the economic activity level with the burden of the health-care system. 

We extend the classic SIR model (\cite{harko2014exact}, \cite{kermack1927contribution}) and incorporate an equilibrium framework to study the optimal lockdown policy during the pandemic period.  What we innovate from previous works is that they all only took the governments' perspective but did not take people's own reaction to the pandemic and the government policy into consideration, while we adopt the extension to the SIR model from  \cite{eichenbaum2020macroeconomics} by involving people's economic decision making (consumption and working hours) and embed the SIR model in a simple Cournot equilibrium framework to model people's reaction to each other. Different from  \cite{eichenbaum2020macroeconomics} that studied the optimal containment policy by controlling the tax rate, we control the level of lockdown, which is more direct and effective for the governments, especially in the early stage of the pandemic. Furthermore, we emphases the cost of death in our model objective of policymakers, which is an important factor in real-world government decision making. Using this method, we can enable the lockdown policy to identify a balance between the impact of the epidemics on the economy and people's health.  

The motivation for this work is to address the following questions that the policymakers may face in reality. The main findings are shown as the short answer to these questions.

\begin{itemize}
    \item What difference does the optimal control make on the economic and health outcome of the epidemic compare to no control? We find that optimal lockdown measures could significantly reduce the deaths and infections caused by the epidemics. Although there is a short-term recession with lockdown control, it has better long-term economic outcomes than doing no control. 
    \item How does the timing of starting and ending affect the optimal lockdown control itself as well as its economic and health consequences? Our results suggest that both the timing of starting and ending the lockdown control policy makes a difference in terms of both the economic and epidemic outcomes. It is best to start the control as early as possible, and it is more important to avoid ending the control too early. 
    \item How does the cost of death affect the lockdown control policy and the outcomes? Whether policymakers regard the deaths as a negative influence on society lead to different results. Regarding deaths as negative results in stricter lockdown control policy which leads to a much better epidemic and slightly worse economic outcomes. 
    \item What if policymakers have additional information on people's health status?  Additional information about the health status of people is beneficial, as the optimal separate control on people in different health status will reach much better economic and epidemic outcomes.

\end{itemize}

This paper continues as follows, in section 2, we describe the SIR-Lockdown Model and analyse its properties. In section 3, we discuss the parameter estimation of the model. We present the numerical results of the SIR-Lockdown model in section 4. Section 5 makes conclusions.

\section{Model}

In this section, we first describe the extension to the canonical SIR model. Then analyse the behaviour of susceptible, infectious, and recovered people in regard to their decisions on consumption and working hours under lockdown regulator and formulate the optimal control problem. Finally, we add the cost of death in our model objective.

\subsection{Extension of SIR}  
\label{section:SIR}

As shown in the classic SIR model  (\cite{kermack1927contribution} \cite{hethcote1989three}),
we classify people into three categories according to \cite{harko2014exact}: 
\begin{itemize}
    \item Infectious  (I) are those who are tested positive to the virus;  
    \item Recovered  (R) are those who have been tested positive to the virus and now recovered;  
    \item Susceptible  (S) are those who have not been tested positive to the virus. 
\end{itemize}

 We assume that all susceptible people are subjects to be infected with some possibility in direct contact with infectious people, and infectious people will recover with a constant probability of $\pi_{r}$ or become dead with another constant probability $\pi_{d}$. Our extension is on the infection. All infection happens via direct contact between susceptible people and infected ones into three types of activities: purchasing and/or consumption of goods and services, working with other people,  and other daily activities. A Lockdown policy can be applied to control the working contact, hence change the income flow, which indirectly imposes constraints on the purchasing and consumption. 

We use the following equation  (1—5) to describe 
our extended SIR model for the transition 
among Susceptible, Infected, Recover, and 
the death outcome. 
\begin{eqnarray}
T_{t}&=&\pi_{s 1}\left (S_{t} C_{t}^{s}\right)\left (I_{t} C_{t}^{i}\right)+\pi_{s 2}\left (S_{t} N_{t}^{s}\right)\left (I_{t} N_{t}^{i}\right)+\pi_{s 3} S_{t} I_{t},   \label{eq1} \\
  S_{t+1}&=&S_t-T_t, \label{eq2} \\
 I_{t+1}&=&I_t+T_t- (\pi_r+\pi_d)I_t, \label{eq3} \\
    R_{t+1}&=&R_t+\pi_r I_t, \label{eq4} \\
D_{t+1}&=&D_t+\pi_dI_t. \label{eq5}
\end{eqnarray}

In this system of equations,  $S_t, I_t, R_t$ 
and $D_t$ represents the number of people 
in categories of Susceptible, Infectious, 
Recovery and Death respectively at time $t$.
We use $ (C^s_t, N^s_t)$ to model the  (average) consumption 
behaviour and working hours of susceptible people, 
$ (C^i_t, N^i_t)$ to model the  (average) consumption behaviour
 and working hours of infectious people, and 
$ (C^r_t, N^r_t)$ to model the  (average) consumption behaviour 
and working hours of a recovered people.  $T_t$ in equation  (\ref{eq1})is the number of 
newly infectious people in the time period $t$ to $t+1$ and the three terms in the right-hand side of this equation are used to describe the infection by the three different contact between susceptible people and infectious people via consumption, working, and other types of contact.  

We use several constant parameters to describe the transition rate between different categories.
$\pi_{s1}$ reflects the transition rate for a susceptible people get infected by infectious people from direct contact via purchasing/consuming.
Similarly, $\pi_{s2}$ reflects the transition rate from direct contact via working, 
and $\pi_{s3}$ reflects the transition rate from other contacts. 

Denote $\Delta Y_t=Y_{t+1}-Y_t$ for $Y=S, I, R$, then the dynamics of the SIR model is 
\begin{eqnarray*}
\Delta S_t&=&-T_t,\\
\Delta I_t&=&T_t- (\pi_r+\pi_d)I_t\\
\Delta R_t&=& \pi_d I_t. 
\end{eqnarray*}

We use vectors and matrices to 
simplify our presentation. Denote $X_t= (S_t, I_t, R_t)^\top, \, C_t= (C_t^s, C_t^i, C_t^r)^\top, n_t= (n^s_t, n^i_t, n^r_t)^\top$, and for any $x= (x_1, x_2, x_3)^\top, c= (c_1, c_2, c_3)^\top, n= (n_1, n_2, n_3)^\top$, define 
\begin{equation}
    T (x, c, n)=x_1x_2\left (\pi_{s 1} c_{1} c_{2}+\pi_{s 2}n_{1} n_{2}+\pi_{s 3} \right)
\end{equation}
\begin{equation}
    F (x,c,n)= (-T (x,c,n), T (x,c,n)- (\pi_r+\pi_d)x_2,\pi_d x_2)^\top
\end{equation}
then the system can be described as 
\begin{equation}
    \Delta X_t=F (X_t, C_t, n_t).
\end{equation}

\subsection{Behaviour of individuals in different categories}

We study the rational behaviour of all people 
who maximise their own welfare by choosing proper 
consumption and working hours like in a normal 
time, i.e., the virus does not change people's 
rationality and preference.  Also, we use the following utility 
function to model the utility from consumption and working
of an individual,
\begin{equation}
u (c,n) = \ln c-\frac{\theta}{2}n^2    
\end{equation}
where $c$ is the consumption, and $n$ is the working 
hours. The first term measures the utility from consumption, and the second term measures the utility from working. 
Denote by $A$ the average wage per hour of a person, hence the labor income 
of an individual, with working hour $n$ is $A*n$, which will be the upper bound of the consumption, i.e. $An\ge c$. 

Denote by $n_0$ the full working hours in a unit time before the spread of the virus, which is officially guided by the government. It is natural that $n_0$ is set optimally for the society, and the optimality brings some information of the parameter $\theta_0$. 
If a person follows the full working hours $n_0$ optimally, then 
her labor income will be $An_0$. 
Since the utility function is strictly increasing in the consumption, all labor income should be consumed up, hence the optimal consumption $c_0$ should be 
Then by the optimality of $n_0$, 
we have $\frac{\partial u (c_0,n_0)}{\partial n}=\frac{1}{n_0}-\theta n_0=0$, by which 
we will choose $\theta$ by 
$$\theta=1/n_0^2.$$

The total utility of a flow of consumption and working 
hours $\{ (c_{\tau}, n_{\tau})\}_{\tau=t,\cdots,T}$ is defined by 
\begin{equation}
U (c_\cdot, n_{\cdot})=\sum_{t=\tau}^{T} \beta^{\tau} u\left (c_{\tau}, n_{\tau}\right)    
\end{equation}

To contain the spreading of the virus, governments need to apply a lockdown policy to reduce direct contacts between people, which will impose stricter
constraints on their behaviour.  In this paper, we study the lockdown policy by a constraint on the ratio $L\in [0, 1]$ of the working hour in the full working capacity, i.e., given the full working hours $n_0$, the maximal working hour cannot exceed $n_0*L$. We suppose the government cannot easily identify individuals into their categories so that the lockdown constraint on the working hours is the same for all people. We formulate the decision making problem for each category with a given lockdown policy $L_\cdot$, and then study the lockdown policy-making problem for the government.

\subsubsection{Optimal decision of recovered people}
Suppose the lockdown measure $L_t\in [0,1]$ is given for any time $t$. 

A recovered people aims at maximising his total utility 
\begin{equation}
    J^r (c^r_\cdot, n^r_\cdot; t)=\sum_{\tau=t}^T \beta^{\tau-t} u (c^r_\tau, n^r_\tau)
\end{equation}
with the constraint $c_\tau^r\le A n_\tau^r$ and $n^r_\tau\le n_0 L_\tau$. 

\begin{theorem}
At time $t$ with state $X_t$ and the lockdown policy $\{L_\tau: \tau\in  [t, T]\}$, 
the optimal $ (c^r, n^r)$ is 
$$c^{r*}_\tau=A n_0 L_\tau, n^{r*}_\tau=n_0 L_\tau, \quad \tau=t, \cdots, T.$$
\end{theorem}  

\begin{proof}
Since $\frac{\partial J^r (c^r,n^r;t)}{\partial c_{\tau}^r}=\beta^{\tau-t}\frac{1}{c_{\tau}^r}>0$, we have $c_{\tau}^{r*}=An_{\tau}^r$ $\forall \tau=t,...,T$. Denote $f (c^r,n^r,\lambda_n^r;t)=J^r (c^r,n^r;t)+\sum_{\tau=t}^T\lambda_{n\tau}^r (n_0L_{\tau}-n_{\tau}^r)$. Then by KKT condition, $\forall \tau=t,...,T$,  
$$\frac{\partial f (c^{r*},n^r,\lambda_n^r;t) }{\partial n_{\tau}^r}=0\Rightarrow \beta^{\tau-t} (-\theta n_{\tau}^r+\frac{1}{n_{\tau}^r})-\lambda_{n\tau}^r=0 $$
$$\lambda_{n\tau}^r (n_0L_{\tau}-n_{\tau}^r)=0,\lambda_{n\tau}^r\ge 0$$
Since $n_0^2\theta=1$, $$\lambda_{n\tau}^r=\beta^{\tau-t} (-\theta n_{\tau}^r+\frac{1}{n_{\tau}^r})>\beta^{\tau-t} (-\theta n_0+\frac{1}{n_0})=0$$
Thus $n^{r*}_\tau=n_0 L_\tau, c_{\tau}^{r*}=An_0L_{\tau}\forall \tau=t,...,T$
\end{proof}

Notice that the behaviour of recovered people $ (c^r_\cdot, n^r_\cdot)$ plays no role in the spread of the virus, hence the behaviour of recovered people will not affect people in other categories. This is why we start to form this easy-to-handle category.

\subsubsection{Optimal behaviour of infectious people}
Similar to the case of recovered people, infectious people also need to choose their optimal consumption and working hours $\{ (c^s_t, n_t^s)\}_{t=0,1,\cdots,T}$ to maximise their total utility from consumption and working hour, 
subject to the constraint that the consumption $c^s_t$ cannot exceed the labour income for the working hour $n_t$, and $n_t$ must be no more than 
the lockdown policy $n_0*L_t$.

The labor income of an infectious people is different from other categories. 
Because they are infected, their health condition is usually worse than other people. So we introduce a constant $\phi$ to discount their working efficiency, and the labor income from $n_t$ working hour will be $A*\phi*n$. Furthermore, since
an infectious people will have a constant probability $\pi_r$ to recover and suffer a possibility $\pi_d$ of death, we need to calculate the distribution over all categories at a future time. 
For an infectious people at time $t$,  he
has the probability $\pi_r$ to recover in the next unit time, $\pi_d$ to die, and the rest probability $1-\pi_r-\pi_d$ to stay in the infected category. 
By this evolution, we can get the conditional probabilities for his health state at a future time $\tau>t$. Denote by $p^{i, i} (t,\tau)$ the probability  for him being  still  infected, $p^{i, r} (t,\tau)$ the probability being recovered, and $p^{i, d} (t,\tau)$
the probability of being dead. Then we can deduce that 
\begin{eqnarray}
p^{i, i} (t,\tau)&=& (1-\pi_r-\pi_d)^{\tau-t},\label{pititau}\\
p^{i, r} (t,\tau)&=&\pi_r\frac{ (1- (1-\pi_r-\pi_d)^{\tau-t})}{\pi_r+\pi_d}, \label{pitrtau}\\
p^{i, d} (t,\tau)&=&\pi_d\frac{ (1- (1-\pi_r-\pi_d)^{\tau-t})}{\pi_r+\pi_d}. \label{pitdtau}
\end{eqnarray}

If he recovered, he should behave optimally as 
a recovered people, while if death has happened  unfortunately, we cease the accumulation of any utility.  So, for a given flow $ (c^i_\cdot, n^i_\cdot)$ of consumption and working hours taken by the infectious people from time $t$, the accumulated utility he can get will be  
\begin{equation}
\begin{aligned}
    J^i (c^i_\cdot, n^i_\cdot; t)=\sum_{\tau=t}^T \beta^{\tau-t} \left[ p^{i,i} (t,\tau) u (c^i_\tau, n^i_\tau)
    -p^{i,r} (t,\tau) u (c^{r*}_\tau, n^{r*}_\tau)\right].
\end{aligned}
\end{equation}
where $ (c^{r*}, n^{r*})$ is the optimal behaviour of a recovered people determined in the previous case, and $p^{i,i}, p^{i,r}$ are as defined in equation  (\ref{pititau}, \ref{pitrtau}). 

\begin{theorem}
Given the lockdown policy $L_\cdot$, 
the optimal $ (c^r, n^r)$ is 
\begin{equation}
    c^{i*}_\tau=A \phi n_0 L_\tau, n^{i*}_\tau=n_0 L_\tau, \quad \tau=t, \cdots, T.
\end{equation}
\end{theorem}

\begin{proof}
Notice $\frac{\partial J^i (c^i,n^i;t)}{\partial c_{\tau}^i}= (\beta (1-\pi_r-\pi_d))^{\tau-t}\frac{1}{c_{\tau}^i}>0$, we have $c_{\tau}^{i*}=A\phi n_{\tau}^i$ $\forall \tau=t,...,T$. Denote $f (c^i,n^i,\lambda_n^i;t)=J^i (c^i,n^i;t)+\sum_{\tau=t}^T\lambda_{n\tau}^i (n_0L_{\tau}-n_{\tau}^i)$. Then by KKT condition, $\forall \tau=t,...,T$,  
$$\frac{\partial f (c^{i*},n^i,\lambda^i;t) }{\partial n_{\tau}^i}=0\Rightarrow  (\beta (1-\pi_r-\pi_d))^{\tau-t} (-\theta n_{\tau}^i+\frac{1}{n_{\tau}^i})-\lambda_{n\tau}^i=0 $$
$$\lambda_{n\tau}^i (n_0L_{\tau}-n_{\tau}^i)=0,\lambda_{n\tau}^i\ge 0$$
Since $n_0^2\theta=1$, $$\lambda_{n\tau}^i= (\beta (1-\pi_r-\pi_d))^{\tau-t} (-\theta n_{\tau}^i+\frac{1}{n_{\tau}^i})> (\beta (1-\pi_r-\pi_d))^{\tau-t} (-\theta n_0+\frac{1}{n_0})=0$$
Thus $n^{i*}_\tau=n_0 L_\tau,c_{\tau}^{i*}=An_0\phi L_{\tau} \forall \tau=t,...,T$

\end{proof}

Different from the infectious case,  the behaviour of an infectious people $ (c^i, n^i)$ is involved in our extend SIR model for the spreading of the virus, 
hence they will make the decision problem for susceptible people much harder.

\subsubsection{Behaviour of susceptible people }

The decision planning for a susceptible people from time $t$ is much more complicated if we consider the possibilities for this people to turn into infectious, recovered, and death at different future time spots. We avoid the complexity by taking advantage of the optimal value function for an infected, and model the objective function of a susceptible people recursively. 

As for the previous two categories, we start from time $t$ and pick up a susceptible person. Denote the state of the SIR model at the starting time as $X_t$, and the lockdown policy is fully given as $L_\cdot)$. 

 Suppose he will  follow a given flow of consumption and working hours $ (c^s_\tau, n^s_\tau)_{\tau=t, t+1,\cdots, T}$ before being infected, and then follow the optimal behaviour  after been infected, i.e., his consumption and working hours after infected will switch to the optimal control for an infected person from the infection time. We denote his objective value as 
 \begin{eqnarray}
 J^s (c^s_\cdot, n^s_\cdot; t, X_t, L_\cdot)&=&u (c^s_t, n^s_t)+\beta \tau_t J^{i*} (t+1, L_\cdot) \nonumber \\
 &&+\beta (1-\tau_t)J^s (c^s_\cdot, n^s_\cdot; t+1, X_{t+1}, L_\cdot),\label{s-recursive}\\
 J^s (c^s_\cdot, n^s_\cdot; T, X_T, L_\cdot)&=&u (c^s_T, n^s_T),
 \end{eqnarray}
 where $\tau_t= \pi_{s1} n_0A \phi I_t L_t c^s_t+\pi_{s2} n_0 I_t L_tn^s_t+\pi_{s3}I_t$ is the probability of a 
 susceptible person to be infected in the next unit time, $J^{i*} (t+1, L_\cdot)$ is the optimal objective value achievable for an infected person starting from time $t+1$, 
 and $X_{t+1}$ is the SIR state at time $t+1$ resulted by people's behaviour $ (c^s_t, n^s_t, c^{i*}_t, n^{i*}_t, c^{r*}_t, n^{r*}_t)$
 and the time $t$ state $X_t$.
 
 Now it is natural that we aim at maximising the objective $J^s (c^s_\cdot, n^s_\cdot; t, X_t, L_\cdot)$ over feasible control flow $ (c^s_\cdot, n^s_\cdot))$, ,i.e., the optimal behaviour of a susceptible people will be the solution for the optimisation 
   \begin{equation}\label{s-optimal}
   \begin{array}{ll}
   \max
   &  J^s (c^s_\cdot, n^s_\cdot; t, X_t, L_\cdot)\\
   s.t.& c_\tau^s\le A n_\tau^s, \quad n^s_\tau\le n_0 L_\tau, \qquad \forall \tau\in \{t, t+1, \cdots, T\}.
   \end{array}
   \end{equation}

\begin{theorem}
At time $t$ with state $X_t$ and the lockdown policy $\{L_\tau: \tau\in  [t, T]\}$, 
the optimal $ (c^s, n^s)$ is 
\begin{equation}
    c^{s*}_\tau=An^{s*}_\tau, \quad \tau=t, \cdots, T.
\end{equation}
\end{theorem}

\begin{proof}
We fix the lockdown policy $L_\cdot$ and omit it when no confusion will arise. 

 Denote the value function as $V (t,X_t)=J^s (c^{s*}_\cdot, n^{s*}_\cdot; t,X_t, L_\cdot)$. 
 According to the dynamic programming principle, we know $V$ must satisfy 
 \begin{eqnarray*}
 V (t,X_t)&=&\max_{c_t^s\le A n_t^s, \, n^s_t\le n_0 L_t} \left[ u (c_t^s, n^s_t)+\beta \tau_t J^{i*} (t+1, L_\cdot)
              +\beta (1-\tau_t) V (t+1, X_{t+1})\right]\\
              &=&u (c_t^{s*}, n^{s*}_t)+\beta \tau^*_t J^{i*} (t+1, L_\cdot)
              +\beta (1-\tau^*_t) V (t+1, X^*_{t+1}),
 \end{eqnarray*}
 where $\tau^*_t$ and $X^*_t$ are the corresponding infection probability and time $t+1$ state of the SIR model.
 
 If $c^{s^*}_t<A n^{s^*}_t$, then, due to the strictly increasing properties of  $\tau_t$ in both $c^{s}$ and $n^s$, 
 we can easily find a value $ m\in  ( c^{s^*}_t, A n^{s^*}_t)$, and construct another  control 
 $c^s_t=m$ and $n^s_t=m/A$, such that the corresponding $\tau_t$ will be the same as $\tau^*_t$, hence $X_{t+1}$
 will also be the same as $X^*_t$. But since $c^s_t>c^{s*}_t$ and $n^s_t <n^{s*}_t$, we have $u (c^s_t, n^s_t)>u (c^{s*}_t, n^{s*}_t)$, which contradicts the optimality of $ (c^{s^*}, n^{s^*})$ in the dynamic programming principle.  
 \end{proof}

\subsection{Optimal Control of the Policymaker}  

With the optimal behaviour in each category under a given lockdown policy $L_\cdot$,  we can easily formulate the optimal policy-making problem into an optimal control problem.

Suppose we start the lockdown problem from some time $t_0$ with the contamination state 
$X_{t_0}$ being given by $S_{t_0}=s, I_{t_0}=i$ 
and $R_{t_0}=r$, then the optimal lockdown policy should be the optimal control problem 

\begin{equation}\label{regulator}
\begin{array}{rl}
\max_{L_\cdot}& J^0 (L_\cdot; t,X_t)=\sum_{t=t_0}^T \beta^{t-t_0} \left [ 
 S_t  u (c^{s*}_t , n^{s*}_t )+ I_t  u (c^{i*}_t , n^{i*}_t )+ R_t  u (c^{r*}_t , n^{r*}_t )  \right],\\
 \end{array} 
\end{equation}
where $ (c^{ca*}_t, n^{ca*}_t)$ are the optimal consumption and working hours for people in category $ca$  ($ca$ can be $s, i$ or $r$), which are all determined in previous optimisation problems.


In previous objective $J^0$,  we remove all cases of death. 
In reality, since death of disease causes has a strong negative impact to a household as well as to the society, regulators should not ignore  any death case. We include the 
strong impact of death cases 
by introduce a penalty term into the objective 
 \begin{equation}\label{death}
 J^\lambda (L_.; t, X_t)=\sum_{\tau=t}^T \beta^{\tau-t} \left [ 
 S_\tau u (c^{s*}_\tau, n^{s*}_\tau)+ I_\tau u (c^{i*}_\tau, n^{i*}_\tau)+ R_\tau u (c^{r*}_\tau, n^{r*}_\tau)  -\lambda D_{\tau}u (c^{r*}_\tau, n^{r*}_\tau)\right],
 \end{equation}
In this new objective,  we measure the the cost of a death by a multiple of the optimal utility for a recovered people,  and the multiple $\lambda>0$ can be viewed as the severity of death in the government's view. 
When $\lambda=0$, $J^\lambda$ reduces to our previous objective $J^0$. 

With this new objective, the  problem for a regulator is to solve 
\begin{equation}\label{regulator1}
\begin{array}{rl}
\max_{L_\cdot}& J^\lambda (L_\cdot; t,X_t),\\
s.t. & L_t\in [0,1] \;\; \forall t\in [0,T].
 \end{array} 
 \end{equation}

\subsection{Solving Scheme}
In Problem  (\ref{regulator1}), or its reduced version  (\ref{regulator}), the optimal decisions of individuals in all three categories are involved. Fortunately,
the optimal decisions of recovered and infectious people are trivial due to our good structure of the model, which leaves us to tackle the optimal decision problem  (\ref{s-optimal}) for susceptible people before the Problem  (\ref{regulator1}). 

We start our solving scheme by tackling the Problem  (\ref{s-optimal}) with a given lockdown policy $L_\cdot$. Because of the lockdown constraint, 
it is almost hopeless for us to get an explicit solution. We solve this optimal control problem numerically in the same was as in  \cite{eichenbaum2020macroeconomics}.
In this approach,  the optimal control at each time step is regarded as the static optimisation with two constraints from the consumption budget and the lockdown policy on the working hours, and solutions are obtained by solving the corresponding KKT condition\footnote{In fact, when we use the numerical scheme proposed in  \cite{eichenbaum2020macroeconomics}
to our problem, the derivative
used in the KKT condition is not correct due to the absence of a complicated term from the term in equation  (\ref{s-recursive}).  We decide to ignore this absence due to the following two reasons:  (1) if we recover this complicated term, the calculation will be extremely complicated;  (2) from real data in the COVID-19 pandemic, we know the coefficient in the third term $\beta \tau_t)$ is very close to $0$, which is also observed in our numerical results. 
}.  

With the optimal control $ (c^{s*}_\cdot, n^{s^*})$ as functions of the lockdown policy $L_\cdot$, we deal with the optimal control problem  (\ref{regulator1}) 
as an optimisation over the high dimension space $[0,1]^T$ by the gradient-based interior-point method used in the Matlab function {\tt fmincon}. Although we have no theoretical proof on the convergence of our scheme, 
our numerical results show the convergence of our scheme. 

Parts of our code in our scheme are from  \cite{eichenbaum2020macroeconomics}.

\section{Model Parameters}  

In this section, we study how to estimate those parameters in our model from real data, and apply it in an  
example with COVID-19 data in the UK to get the numerical results for optimal lockdown control.

In our model, we have quite a lot of parameters, and some of them are well-estimated and available from different sources. 
Let us start from easily accessible ones. 

For the extended SIR model, without loss of generality, we standardise  the total population to $N=1$, which makes $S_t, I_t, R_t$ and $D_t$ be the 
proportions of the population of each category in the total population. 

The unit of a time step is not an essential parameter,  we can simply count the time by weeks. 

$\pi_r$ and $\pi_d$ in the extended SIR model can be easily estimated from historical data, which have been done in several data sources like HPCC covid19 data cluster \footnotemark
. In our example, we will use the estimation from  \cite{eichenbaum2020macroeconomics}. 

$\pi_{s1}, \pi_{s2}$ and $\pi_{s3}$ are complicated to estimate, and we defer the discussion to after all easy ones. 
 
For the characterisation of the decision making for individuals, we still need parameters
$n_0, \theta, \beta, A$, and $\phi$. Most of them are quite flexible, and in our examples, we do not estimate them from real data but specify their values in the same way as in different literature. 
We will do it in our detailed example. 

Finally, let us focus on the estimation of $\pi_{s_1},\pi_{s_2},\pi_{s_3}$. 
At any time $t$, we have $\pi_{s1}c_t^sc_t^i+\pi_{s2}n_t^sn_t^i+\pi_{s3}=\pi_{t}$, where $\pi_{t}$ is the transmission rate in classic SIR model. Similar to $\pi_r$ and 
$\pi_d$, the quantity $\pi_t$ is also available in different data source like HPCC covid-19 data cluster\footnotemark[\value{footnote}]
\footnotetext{HPCC systems covid19: \url{https://covid19.hpccsystems.com/}}.
To estimate $\pi_{s1},\pi_{s2},\pi_{s3}$, we choose two different time spots $t_1$ and $t_2$. The first time spot  $t_1$ can be any time between  the onset of the spreading of the virus and the first lockdown measure, and the second time spot $t_2$ must be in a period where  a lockdown measure was applied. With the observation of $\pi_{t_1}$ and $\pi_{t_2}$, 
 we have: 

$$\pi_{s1}A^2n_0^2+\pi_{s2}n_0^2+\pi_{s3}=\pi_{t_1}, \quad  
\pi_{s1}A^2n_0^2L_{t_2}^2+\pi_{s2}n_0^2L_{t_2}^2+\pi_{s3}=\pi_{t_2},$$  
where $L_t$ is an estimation of actual lockdown rate at time $t$. 
These two equations are not enough to give us the values of three parameters, we still need one more equation for the purpose. 
In the case  (as happened in the UK)  that no different  (non-null) lockdown measures have been applied, the third equation is officially not available. So we assume that 
$$\pi_{s2}n_0^2\times \frac{1}{3}\times\frac{1}{6}=\pi_{s3}.$$ 
This equation is from the assumption that susceptible people spend about $1/3$ of their working hours for other activities related to other types of direct contact, 
and infectious people spend about half the time of susceptible ones in this type of activity due to the poor health condition.  The two proportions $1/3$ and $1/6$ can be adjusted based on personal experience. 


These three equations can give us a good estimation of $\pi_{s1},\pi_{s2}$ and $\pi_{s3}$.

\subsection{Parameter estimation: an example}
We take the COVID-19 in the UK as our example, which 
started in the year 2019. The only lockdown took place on 23 March 2020 and lifted up in July 2020. 

For the estimation of $\pi_{s1},\pi_{s2}$ and $\pi_{s3}$, we need to specify some other parameters. 

According to the starting of the epidemic and lockdown,
we take $t_1$ to be a time in Jan 2020 
and $t_2$ to be some time in April 2020.  

The government released Experimental results of the pilot Office for National Statistics  (ONS) online time-use study  (collected 28 March to 26 April 2020 across Great Britain) \footnote{ONS Dataset \url{https://www.ons.gov.uk/economy/nationalaccounts/satelliteaccounts/datasets/coronavirusandhowpeoplespenttheirtimeunderlockdown}}
compared with the 2014 to 2015 UK time-use study, which reported the working-not-from-home time. According to the study, the average daily time  (in minutes) of working not from home is $97.6$ in March/April 2020 and $150.0$ in 2014/2015, thus we estimate   
$L_{t2} = 97.6/150 \approx 0.65$.  


Also according to ONS, the average actual weekly hours of work for full-time workers from Dec 2019 to Feb 2020 was $36.9$ \footnote{ ONS,
Average actual weekly hours of work for full-time workers:\url{https://www.ons.gov.uk/employmentandlabourmarket/peopleinwork/earningsandworkinghours/timeseries/ybuy/lms}} , thereby we set $n_0=36.9$.  According to the equation $n_0^2\theta=1$, we set $\theta=0.00073$.

We follows the setting of some parameters in literature. The mortality rate is set to be  $0.6\%$ from  \cite{eichenbaum2020macroeconomics}. As in  \cite{eichenbaum2020macroeconomics} , we assume that each infected case takes $18$ days on average to either recover or die. Since our model is weekly, we have $\pi_d=0.006\times 7/18,\pi_r=7/18-\pi_d$.  
The reproduction number $R_0$ 
at time $t_1$ in Jan 2020 is around $1.95$ without control measures \footnote{Coronavirus wikipedia: \url{https://en.wikipedia.org/wiki/Coronavirus_disease_2019}}, and between $0.7$ to $1.0$ in April 2020 after the lockdown \footnote{BBC report on R number: \url{https://www.bbc.co.uk/news/health-52677194}}\footnote{The R number in the UK:\url{https://www.gov.uk/guidance/the-r-number-in-the-uk}}, we use the middle point $0.85$ of this range of $R_0$ for the calcuation of $\pi_{t_2}$. Since in classic SIR model, $R_0=\beta/\gamma$ where $\beta$ and $\gamma$ the infected and recovery transmission rate
$$\pi_{t_1}=1.95\times 7/18,\pi_{t_2}=0.85\times 7/18.$$
Given a published  average annual income \footnote{statista: average full time annual earnings in the uk: \url{https://www.statista.com/statistics/1002964/average-full-time-annual-earnings-in-the-uk/}} $30350$ for $52$ weeks, we set $A=15.8172$.  

With all quantities  involved in the three equations for  $ (\pi_{s1},\pi_{s2},\pi_{s3})$, we get 
 solution 
$$\pi_{s1}=1.244887\times 10^{-6},\pi_{s2}=1.0336\times 10^{-4}, \pi_{s3}=0.01759.$$ 

By the value $n_0=36.9$, we take 
$\theta=1/ (36.9)^2$. 

Finally, we copy the value $\phi=0.8$ from   \cite{eichenbaum2020macroeconomics}. 

\section{Numerical Results} 
In this section, we present the result of our numerical experiments under the parameter setting in section 3.2.   
We do experiments to analyze the impact of the optimal lockdown control policy, the policy when different levels of the cost of death are taken into consideration, early exit and late start of the lockdown policy, and finally the smart containment policy. For every experiment, the initial state is $ (S,I,R)= (0.9998,0.0002,0)$ and the time horizon is $100$ weeks.

\subsection{Optimal Lockdown Control}

As Figure 1 (a), (d) (page \pageref{fig1}) shows, if there is no lockdown control, i.e. the lockdown rate is constant $1$ for all time, then under our parameter setting, around $15\%$ of the population will be infected, $0.3\%$ of the population will die and the peak of infection will be above $0.6\%$ at week $50$. Under the optimal lockdown control, the proportion of Infectious people decrease to 
$5.22\times 10^{-5}$ at week $50$, then raises to $2.5\times 10^{-4}$ at week $100$. $0.37\%$ of the population will become infected and $0.0068\%$ of the population will die by week $100$. The optimal lockdown policy reduces the peak of infection by $95.8\%$ and reduces the number of deaths by $97.7\%$.  
The significant life-saving is associated with a recession. Figure 1 (e) (page \pageref{fig1}) shows the aggregate consumption under optimal lockdown policy decreases $20\%$ compares to the no control case at the beginning, but then constantly increases. The average aggregated consumption fall by $6.6\%$ with the optimal lockdown measure. 
In Figure 1 (f) (page \pageref{fig1}), the optimal lockdown rate starts from around $80\%$, then gradually release to above $95\%$, the speed of the increase of the lockdown rate first decreases until around week $50$, then increase until week $100$.  

The increase of the infected proportion is because our model has a finite time horizon, and does not take the consequences after a time horizon of $100$ weeks into consideration. In the beginning, the aggregated consumption under optimal lockdown control is $20\%$ less but becomes $8.2\%$ more than that of no control in the end. The reason that the optimal lockdown control policy did not cause a severe recession might be that in the no control case, susceptible people will cut back their working hours, as well as their consumption as the infected population increases, and in the optimal lockdown control restricted the infected population so that susceptible people won't cut back their consumption as much. We proved in section 2 that the recovered and Infectious people will work as much time as possible in order to maximize their own utility, but the behaviour of susceptible people is not certain. In the parameter setting of our experiments, the susceptible people almost work as much as possible just as the infected and recovered people do, but slightly reduce their working hours from the upper bound of lockdown constrain near the end of the time horizon, this behaviour may due to the increase of infected proportion, which raises the risk of getting infected for susceptible people.  

In general, the optimal lockdown policy saves lives and is more robust in economic recovery, it brings long-term health benefits and economic growth with the cost of a short-term recession.

\subsection{Cost of Death}  

In this subsection, we study how the severity of death regarded by the planners affects the optimal lockdown policy. We set the penalty coefficient of death $\lambda$ in $ (17)$ as $0,10,20,50$, which means the death of $1$ people is regarded as the loss of $0,10,20,50$ recovered people by the planner. When $\lambda=0$, it is the same as the original optimal control model.   

Our results in Figure 2 (page \pageref{fig2}) show that adding a penalty on deaths makes a huge difference, it significantly slows down the increase of the lockdown rate (Figure 2 (f) (page \pageref{fig2})), thus reduces the proportion of deaths in a great extent: $76.7\%,83.0\%,87.2\%$ respectively (Figure 2 (d) (page \pageref{fig2})), and avoid the substantial rise of the infectious population (Figure 2 (a) (page \pageref{fig2})), these are beneficial in terms of the mental impact in the society as low deaths and infection amount release the pressure on both people in the society and the planner. As the penalty coefficient increases, the optimal policy becomes constantly more strict. The relation of the death penalty coefficient and the result optimal control rate is below linear. As Figure 2 (f) (page \pageref{fig2}) shows, despite that optimal lockdown policy with different death penalty coefficient starts with quite different lockdown rates:$0.62,0.56,0.46$ for penalty coefficient $10,20,50$ respectively, they quickly become close. At the end of control, the aggregate consumption, as well as the lockdown rate of optimal policy with the death penalty is extremely close to the one without the death penalty. Compare to the original optimal lockdown policy, there is a slight recession when adding penalty on number of deaths: the average aggregate consumption decreases by $3.3\%,4.5\%,5.5\%$ for penalty coefficient $10,20,50$ respectively (Figure 2 (e) (page \pageref{fig2})).  

Although the lockdown control policy with or without the death penalty becomes close from the middle to the end of the control, in the latter case, the infectious population does not rise as it in the former case. This is because that the lockdown policy with the death penalty suppresses the infectious population to a much lower level than the lockdown policy without the death penalty, thus the infectious population grows slower as the lockdown rate increases.    

In general, considering the cost of deaths leads to a more conservative lockdown control policy, it saves much more lives at the cost of a short-term recession.

\subsection{Cost of Early Ending of Lockdown Control Policy}  

Practically, policymakers may under the intense pressure of economic loss that forces them to end the containment policy in the middle of the pandemic. In this subsection, we discuss the consequences of doing so. As we see in section 4.1, the infected population reaches the bottom at the week $50$, which may seem to be a good time spot to end the lockdown policy.  

Our results in Figure 3 (e) (page \pageref{fig3}) shows that there is an instant bounce of consumption right after the end of lockdown control, but this would cause the instant rise of infectious population (Figure 3 (a) (page \pageref{fig3})), and at the end, the infectious population $72$ times larger than that of week $50$. The burst of infection would result in a recession of $10.8\%$ from the peak at the end (Figure 3 (e) (page \pageref{fig3})).  

So, ending the lockdown policy prematurely may not bring long-term economic benefit and what's worse is, it would result in a substantial additional number of deaths. Therefore we suggest that policymakers avoid terminating the lockdown policy during the pandemic in pursuit of only a short-term economic benefit.

\subsection{Cost of Start the Lockdown Control Policy Late}  

Policymakers could also face the situation that there are things that prevent them from taking the lockdown measure in the early stage of the pandemic.   

Our results Figure 4 (f) (page \pageref{fig4}) show the optimal lockdown policy that starts at week $13$ (around $3$ months later). Compare to the optimal lockdown policy that starts at week $0$ that starts with the lockdown rate $0.8$, the late started optimal lockdown policy starts with a stricter constrain rate of $0.73$. Although the lockdown rate of the late started lockdown policy constantly increases, it is always less than the original optimal lockdown policy. The late start causes a slight stronger recession (Figure 4 (e) (page \pageref{fig4})): the average aggregated consumption reduces $1\%$ and a substantial rise of deaths (Figure 4 (d) (page \pageref{fig4})): the number of deaths rises $84.8\%$ by week $100$.    

In general, It is the earlier the better to start the lockdown control policy, and despite that the late start of lockdown policy brings additional loss, it is much better than applying no containment policy or abandon it too early.

\subsection{Vaccination}  

Vaccination is an effective method of preventing infectious diseases. We now involve vaccination in SIR model. Assume that at each time period, fix amount of susceptible people: $\delta_{v}$ of the starting population get vaccination that could prevent them from getting COVID-19 and assume governments to afford the cost of vaccination for people. Once susceptible people get vaccination, they are regarded as recovered. Thus the objective value of susceptible people become: 

\begin{eqnarray}
 J^s (c^s_\cdot, n^s_\cdot; t, X_t, L_\cdot)&=&u (c^s_t, n^s_t)+\beta(1-\frac{\delta_v}{S_t})\tau_t J^{i*} (t+1, L_\cdot) \nonumber \\
 &&+\beta(1-\frac{\delta_v}{S_t}) (1-\tau_t)J^s (c^s_\cdot, n^s_\cdot; t+1, X_{t+1}, L_\cdot)\nonumber\\
 &&+\beta\frac{\delta_v}{S_t} J^{r*} (t+1, L_\cdot)
 \end{eqnarray}

With the cost of vaccination, denote as $p$, the optimal control problem of policymakers become: 
\begin{equation}
\begin{array}{rl}
\max_{L_\cdot}& J^0 (L_\cdot; t,X_t)=\sum_{t=t_0}^T \beta^{t-t_0} \left [ 
 S_t  u (c^{s*}_t , n^{s*}_t )+ I_t  u (c^{i*}_t , n^{i*}_t )+ R_t  u (c^{r*}_t , n^{r*}_t )- p\delta_v \right]\\
 \end{array} 
\end{equation}
We set $\delta_v=1/104$ in simulation. 
Results on Figure 7(a),(b),(d)(page \pageref{fig7}) shows that vaccination could eliminate the epidemic without a rebound of infection and reduce the number of deaths compare that without vaccination. Figure(e),(f) shows that vaccination reduces the severity of recession and leads to a less strict optimal lockdown control.

\subsection{Smart Lockdown Control Policy}

In the lockdown control policies we studied so far, the government chooses the same lockdown rate for all three kinds of people (susceptible, infectious, and recovered). In this subsection, we consider the smart containment, by which means the policymaker directly chooses working hours for all three kinds of people with the same objective function as previous models.    
There is no need to apply any lockdown on recovered people because their utility reaches the maximum as their working hour is at the maximum and they do not affect the utility or the transition of susceptible and Infectious people. Our results show that in the smart lockdown control policy, Infectious people almost do not work at the beginning, but then the planner gradually increases their working hours as the infected population decreases rapidly, and susceptible people work almost without fear of becoming infected.  
Figure 5 (page \pageref{fig5}) shows that compare to the previous optimal lockdown control policy, the smart lockdown policy is much better, since it reduces the number of deaths to a great extent, and almost avoids the recession because the proportion of Infectious people is extremely small.  
The implement of a smart lockdown control policy requires the planners to know the status of all people and have control over their working hours. In reality, the knowledge of people's status needs measures such as medical testing and rely on the accuracy of testing. Our results suggest that these measures and information that are helpful for taking smart lockdown policy are beneficial for social welfare.

\subsection{View of Reproduction Number}  

The reproduction number ($R_0$) is now a basis for some governments to make decisions in reaction to the pandemic. We present the $R_0$ of lockdown polices in all our experiments in Figure 6 (page \pageref{fig6}).  
The $R_0$ of smart lockdown policy is much smaller than that of all other lockdown policies. The $R_0$ of lockdown policies with the same lockdown rate for all three kinds of people behave similarly to their lockdown rate whereas in the no control case, its $R_0$ decreases constantly and the $R_0$ of smart lockdown policy behaves similar to the lockdown rate of Infectious people. This is because the behaviour of all three kinds of people is in accordance with the lockdown rate in optimal lockdown control policies, while the behaviour of susceptible people varies if there is no control, and in the smart lockdown case, the susceptible and recovered people almost remain the lockdown rate as constant $1$ in the whole control process, thus its $R_0$ behaves in accordance with the lockdown rate of Infectious people. Notice that although the $R_0$ of in the no control case decreases below $1$, and the $R_0$ of lockdown control policies with or without the death penalty increase over $1$, policies with lockdown control are much better than that without control as analyzed in previous subsections. We, therefore, suggest that whether $R_0$ is larger or less than $1$ can not be the only foundation for planners to make judgements or decisions on the current situation.

\section{Conclusion}  

In this paper, we extend the canonical epidemiological model SIR to find an optimal decision making with the aim to balance between economy and people's health. In our model, people in different health statuses take different decisions on their working hours and consumption to maximise their own utility, while policymakers control the lockdown rate to maximise the overall welfare, which leads to a two phases optimisation problem. Several parameters in our model are not straightforward to specify using the common epidemic data for modelling. We develop a novel method of parameter estimation through various additional sources of data. Our results show that lockdown measures could effectively reduce the deaths and infections caused by the COVID-19. There is an inevitable trade-off between the short-term recession, and health problems caused by the pandemic, and how policymakers deal with this could lead to very different decisions. We quantify the trade-off by emphasising the cost of death in the model objective, which enables the optimal lockdown policy to discover a balance between the economic and epidemic outcomes. The timing of starting and ending the lockdown control policy makes much difference in terms of both the economic and epidemic outcomes. So the earlier to start the control, the better the results will be. It is crucial to avoid premature ending of the control. In the analysis of the smart containment policy, the results suggest that additional information about the health status of people is beneficial, as the optimal lockdown control policy will reach much better outcomes if it could be implemented on people with different health status separately. Through comparison of lockdown policies, we suggest that $R_0$ cannot be the only foundation for policy-making.


\begin{figure}
  \centering
  \includegraphics[width=0.9\textwidth]{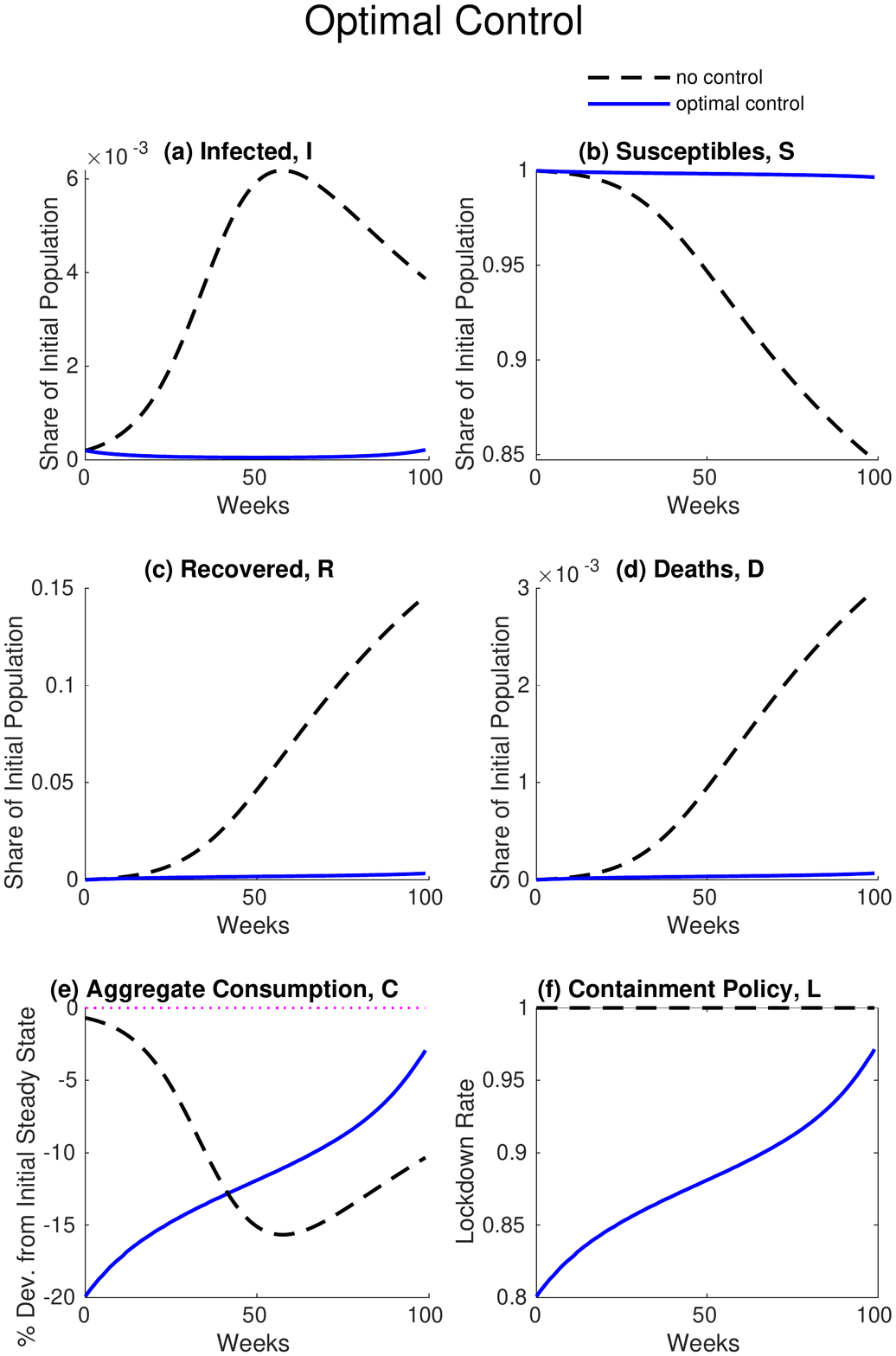}\label{fig1}
\end{figure}


\begin{figure}
  \centering
  \includegraphics[width=0.9\textwidth]{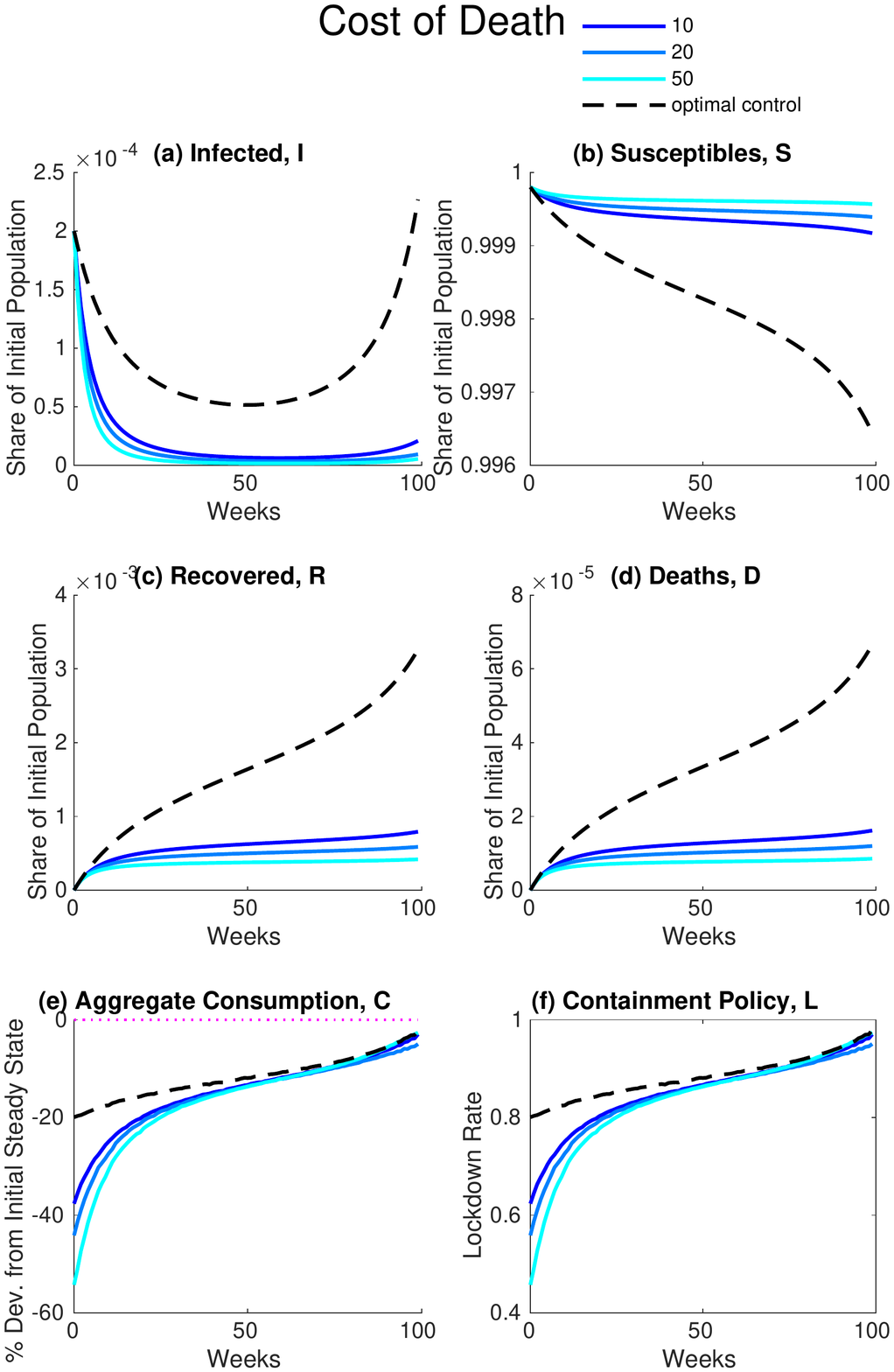}\label{fig2}
\end{figure}


\begin{figure}
  \centering
  \includegraphics[width=0.9\textwidth]{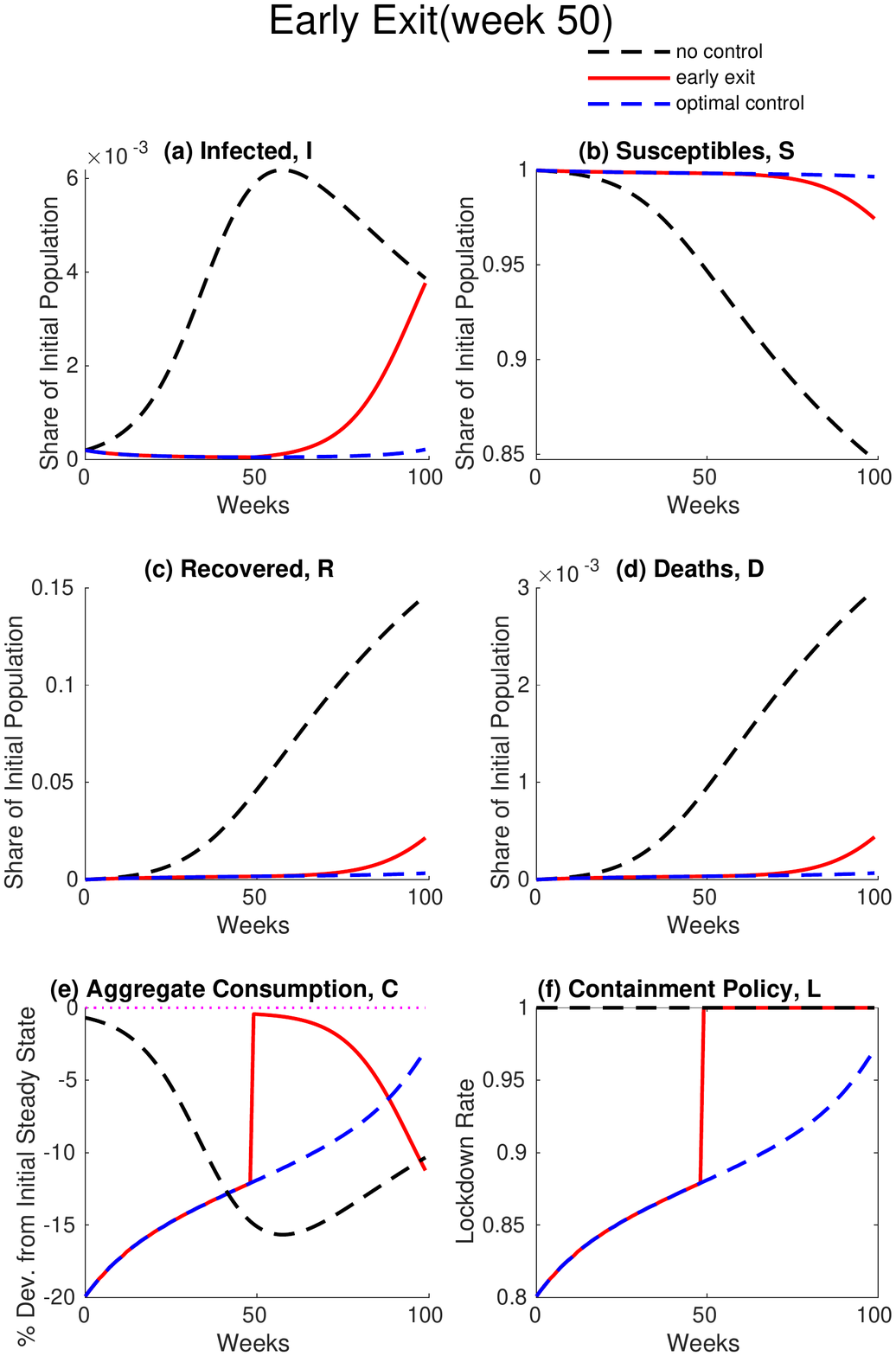}\label{fig3}
\end{figure}


\begin{figure}
  \centering
  \includegraphics[width=0.9\textwidth]{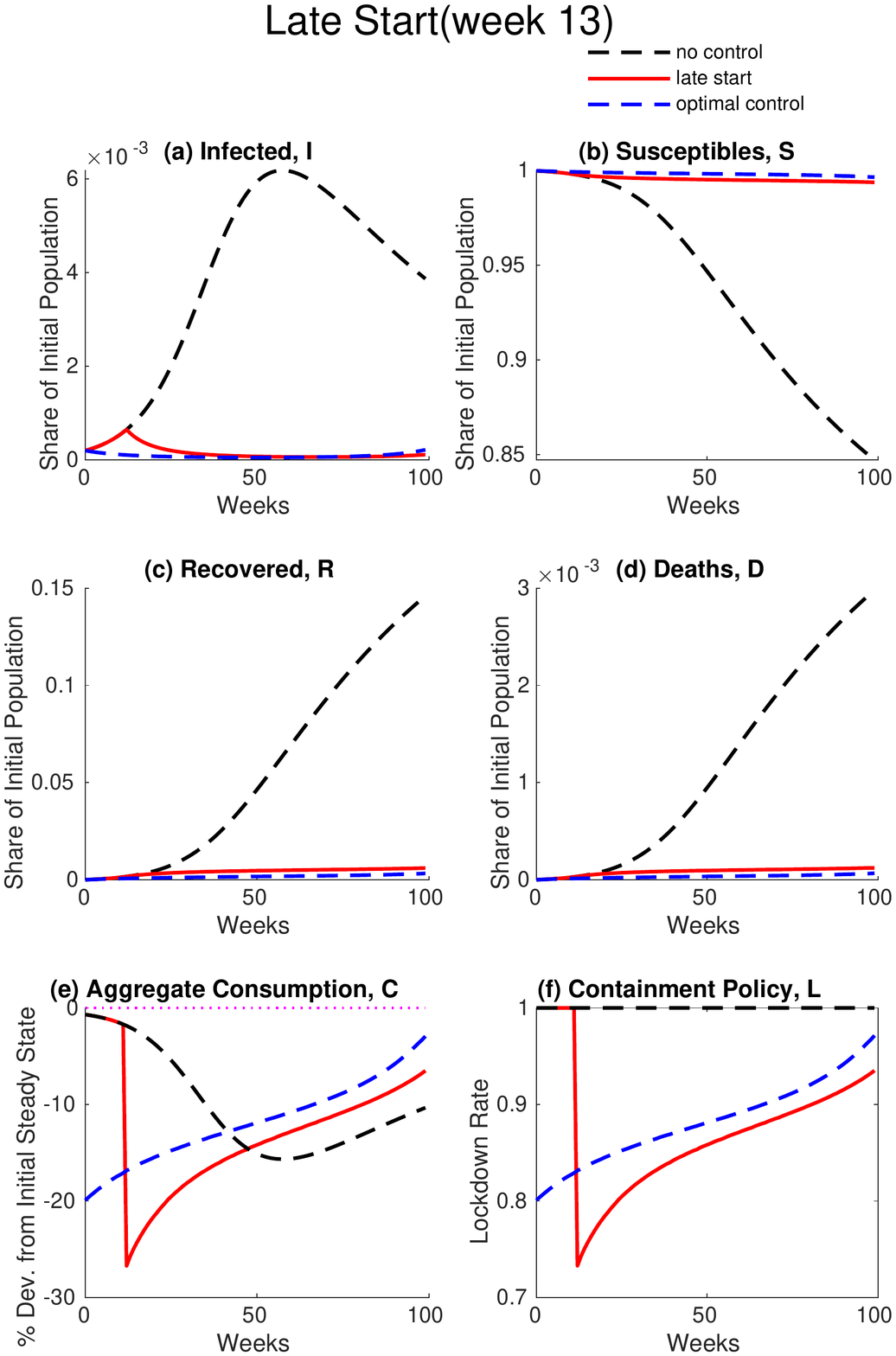}\label{fig4}
\end{figure}


\begin{figure}
  \centering
  \includegraphics[width=0.9\textwidth]{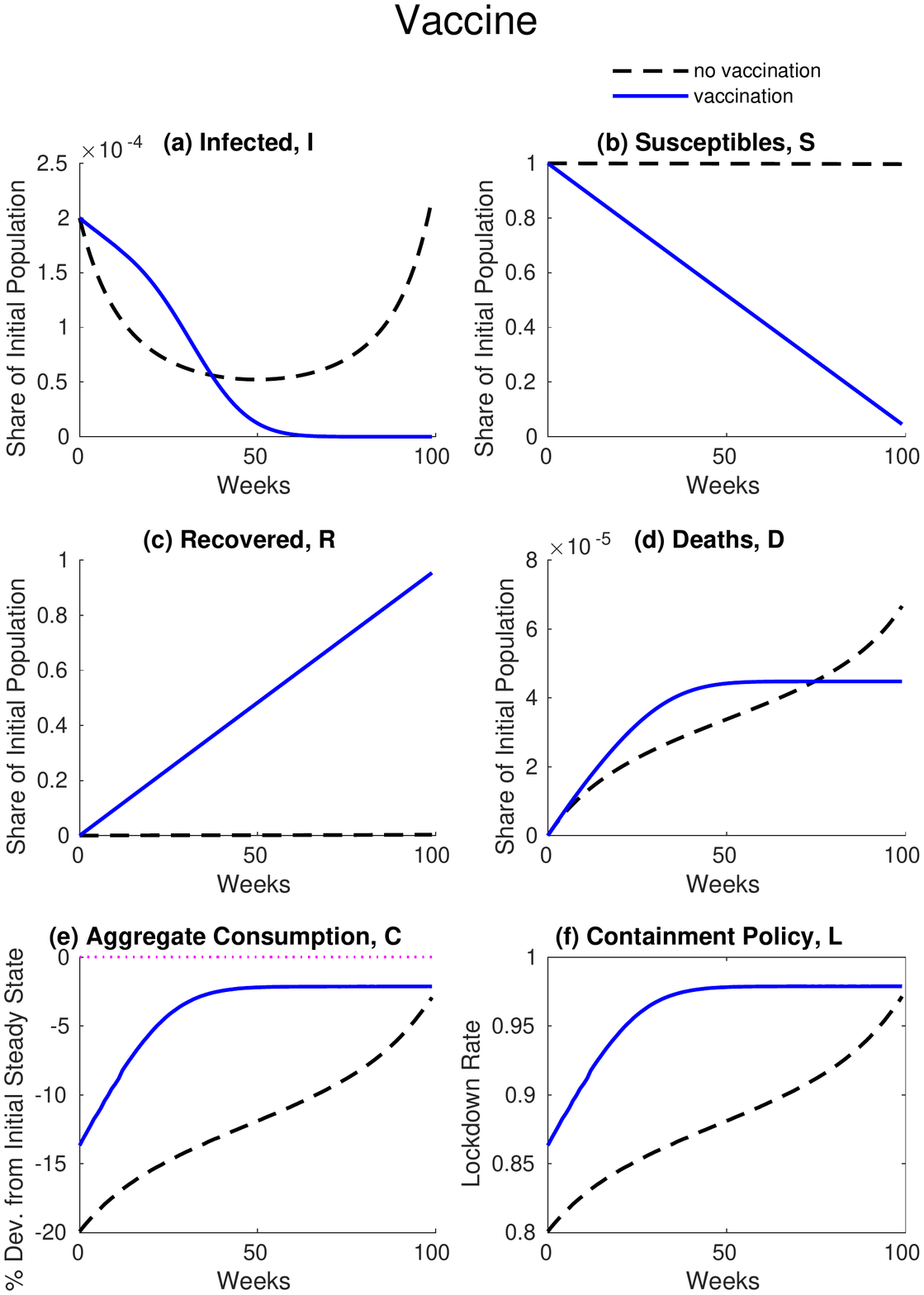}\label{fig5}
\end{figure}


\begin{figure}
  \centering
  \includegraphics[width=0.9\textwidth]{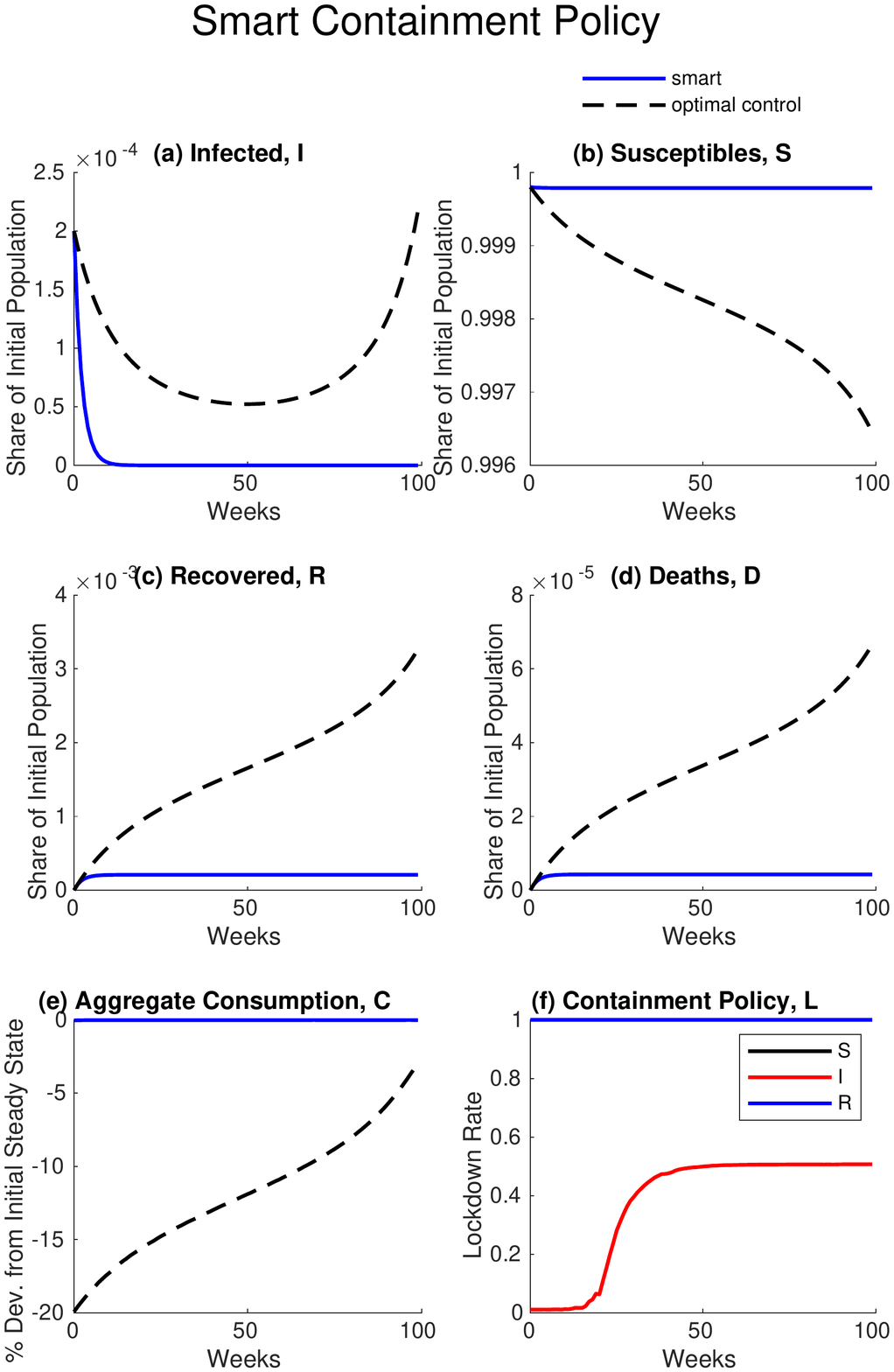}\label{fig6}
\end{figure}


\begin{figure}
  \centering
  \includegraphics[width=0.9\textwidth]{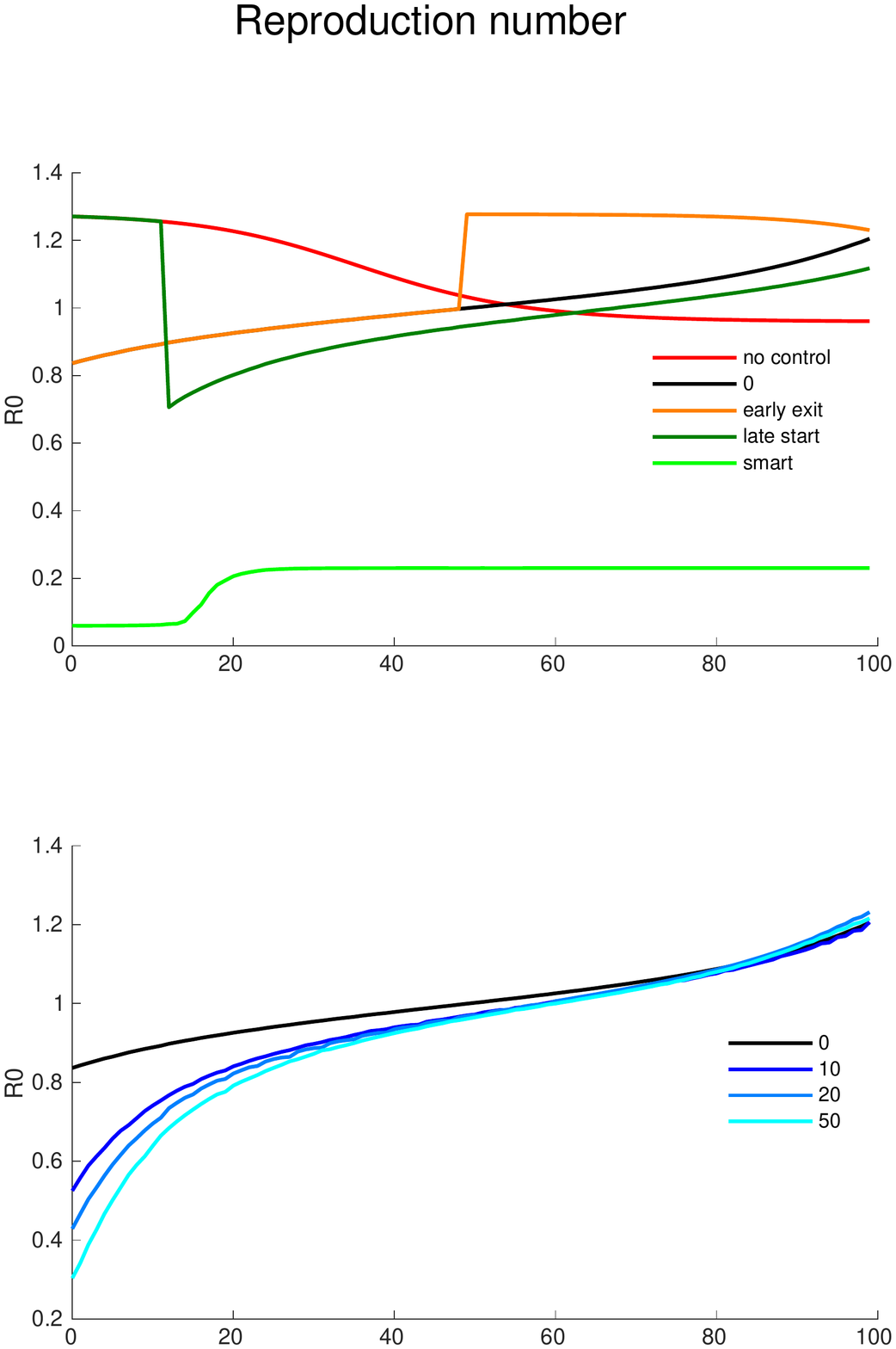}\label{fig7}
\end{figure}

\bibliographystyle{elsarticle-harv}
\bibliography{reference.bib}

\begin{thebibliography}{26}
\expandafter\ifx\csname natexlab\endcsname\relax\def\natexlab#1{#1}\fi
\providecommand{\url}[1]{\texttt{#1}}
\providecommand{\href}[2]{#2}
\providecommand{\path}[1]{#1}
\providecommand{\DOIprefix}{doi:}
\providecommand{\ArXivprefix}{arXiv:}
\providecommand{\URLprefix}{URL: }
\providecommand{\Pubmedprefix}{pmid:}
\providecommand{\doi}[1]{\href{http://dx.doi.org/#1}{\path{#1}}}
\providecommand{\Pubmed}[1]{\href{pmid:#1}{\path{#1}}}
\providecommand{\bibinfo}[2]{#2}
\ifx\xfnm\relax \def\xfnm[#1]{\unskip,\space#1}\fi
\bibitem[{Acemoglu et~al.(2020)Acemoglu, Chernozhukov, Werning and
  Whinston}]{acemoglu2020optimal}
\bibinfo{author}{Acemoglu, D.}, \bibinfo{author}{Chernozhukov, V.},
  \bibinfo{author}{Werning, I.}, \bibinfo{author}{Whinston, M.D.},
  \bibinfo{year}{2020}.
\newblock \bibinfo{title}{Optimal targeted lockdowns in a multi-group sir
  model}.
\newblock \bibinfo{journal}{NBER Working Paper} \bibinfo{volume}{27102}.
\bibitem[{Aleta et~al.(2020)Aleta, Mart{\'\i}n-Corral, y~Piontti, Ajelli,
  Litvinova, Chinazzi, Dean, Halloran, Longini~Jr, Merler
  et~al.}]{aleta2020modelling}
\bibinfo{author}{Aleta, A.}, \bibinfo{author}{Mart{\'\i}n-Corral, D.},
  \bibinfo{author}{y~Piontti, A.P.}, \bibinfo{author}{Ajelli, M.},
  \bibinfo{author}{Litvinova, M.}, \bibinfo{author}{Chinazzi, M.},
  \bibinfo{author}{Dean, N.E.}, \bibinfo{author}{Halloran, M.E.},
  \bibinfo{author}{Longini~Jr, I.M.}, \bibinfo{author}{Merler, S.}, et~al.,
  \bibinfo{year}{2020}.
\newblock \bibinfo{title}{Modelling the impact of testing, contact tracing and
  household quarantine on second waves of covid-19}.
\newblock \bibinfo{journal}{Nature Human Behaviour} \bibinfo{volume}{4},
  \bibinfo{pages}{964--971}.
\bibitem[{Alvarez et~al.(2020)Alvarez, Argente and Lippi}]{alvarez2020simple}
\bibinfo{author}{Alvarez, F.E.}, \bibinfo{author}{Argente, D.},
  \bibinfo{author}{Lippi, F.}, \bibinfo{year}{2020}.
\newblock \bibinfo{title}{A simple planning problem for covid-19 lockdown}.
\newblock \bibinfo{type}{Technical Report}. National Bureau of Economic
  Research.
\bibitem[{Berger et~al.(2020)Berger, Herkenhoff and Mongey}]{berger2020seir}
\bibinfo{author}{Berger, D.W.}, \bibinfo{author}{Herkenhoff, K.F.},
  \bibinfo{author}{Mongey, S.}, \bibinfo{year}{2020}.
\newblock \bibinfo{title}{An seir infectious disease model with testing and
  conditional quarantine}.
\newblock \bibinfo{type}{Technical Report}. National Bureau of Economic
  Research.
\bibitem[{Brotherhood et~al.(2020)Brotherhood, Kircher, Santos and
  Tertilt}]{brotherhood2020economic}
\bibinfo{author}{Brotherhood, L.}, \bibinfo{author}{Kircher, P.},
  \bibinfo{author}{Santos, C.}, \bibinfo{author}{Tertilt, M.},
  \bibinfo{year}{2020}.
\newblock \bibinfo{title}{An economic model of the covid-19 epidemic: The
  importance of testing and age-specific policies} .
\bibitem[{Faria-e Castro(2020)}]{faria2020fiscal}
\bibinfo{author}{Faria-e Castro, M.}, \bibinfo{year}{2020}.
\newblock \bibinfo{title}{Fiscal policy during a pandemic}.
\newblock \bibinfo{journal}{FRB St. Louis Working Paper} .
\bibitem[{Eichenbaum et~al.(2020)Eichenbaum, Rebelo and
  Trabandt}]{eichenbaum2020macroeconomics}
\bibinfo{author}{Eichenbaum, M.S.}, \bibinfo{author}{Rebelo, S.},
  \bibinfo{author}{Trabandt, M.}, \bibinfo{year}{2020}.
\newblock \bibinfo{title}{The macroeconomics of epidemics}.
\newblock \bibinfo{type}{Technical Report}. National Bureau of Economic
  Research.
\bibitem[{Farboodi et~al.(2020)Farboodi, Jarosch and
  Shimer}]{farboodi2020internal}
\bibinfo{author}{Farboodi, M.}, \bibinfo{author}{Jarosch, G.},
  \bibinfo{author}{Shimer, R.}, \bibinfo{year}{2020}.
\newblock \bibinfo{title}{Internal and external effects of social distancing in
  a pandemic}.
\newblock \bibinfo{type}{Technical Report}. National Bureau of Economic
  Research.
\bibitem[{Ferguson et~al.(2005)Ferguson, Cummings, Cauchemez, Fraser, Riley,
  Meeyai, Iamsirithaworn and Burke}]{ferguson2005strategies}
\bibinfo{author}{Ferguson, N.M.}, \bibinfo{author}{Cummings, D.A.},
  \bibinfo{author}{Cauchemez, S.}, \bibinfo{author}{Fraser, C.},
  \bibinfo{author}{Riley, S.}, \bibinfo{author}{Meeyai, A.},
  \bibinfo{author}{Iamsirithaworn, S.}, \bibinfo{author}{Burke, D.S.},
  \bibinfo{year}{2005}.
\newblock \bibinfo{title}{Strategies for containing an emerging influenza
  pandemic in southeast asia}.
\newblock \bibinfo{journal}{Nature} \bibinfo{volume}{437},
  \bibinfo{pages}{209--214}.
\bibitem[{Ferguson et~al.(2020)Ferguson, Laydon, Nedjati-Gilani, Imai, Ainslie,
  Baguelin, Bhatia, Boonyasiri, Cucunub{\'a}, Cuomo-Dannenburg
  et~al.}]{ferguson2020impact}
\bibinfo{author}{Ferguson, N.M.}, \bibinfo{author}{Laydon, D.},
  \bibinfo{author}{Nedjati-Gilani, G.}, \bibinfo{author}{Imai, N.},
  \bibinfo{author}{Ainslie, K.}, \bibinfo{author}{Baguelin, M.},
  \bibinfo{author}{Bhatia, S.}, \bibinfo{author}{Boonyasiri, A.},
  \bibinfo{author}{Cucunub{\'a}, Z.}, \bibinfo{author}{Cuomo-Dannenburg, G.},
  et~al., \bibinfo{year}{2020}.
\newblock \bibinfo{title}{Impact of non-pharmaceutical interventions (npis) to
  reduce covid-19 mortality and healthcare demand. 2020}.
\newblock \bibinfo{journal}{DOI} \bibinfo{volume}{10}, \bibinfo{pages}{77482}.
\bibitem[{Flaxman et~al.(2020)Flaxman, Mishra, Gandy, Unwin, Mellan, Coupland,
  Whittaker, Zhu, Berah, Eaton et~al.}]{flaxman2020estimating}
\bibinfo{author}{Flaxman, S.}, \bibinfo{author}{Mishra, S.},
  \bibinfo{author}{Gandy, A.}, \bibinfo{author}{Unwin, H.J.T.},
  \bibinfo{author}{Mellan, T.A.}, \bibinfo{author}{Coupland, H.},
  \bibinfo{author}{Whittaker, C.}, \bibinfo{author}{Zhu, H.},
  \bibinfo{author}{Berah, T.}, \bibinfo{author}{Eaton, J.W.}, et~al.,
  \bibinfo{year}{2020}.
\newblock \bibinfo{title}{Estimating the effects of non-pharmaceutical
  interventions on covid-19 in europe}.
\newblock \bibinfo{journal}{Nature} \bibinfo{volume}{584},
  \bibinfo{pages}{257--261}.
\bibitem[{Gollier(2020)}]{gollier2020cost}
\bibinfo{author}{Gollier, C.}, \bibinfo{year}{2020}.
\newblock \bibinfo{title}{Cost-benefit analysis of age-specific deconfinement
  strategies}.
\newblock \bibinfo{journal}{URL: https://drive. google.
  com/file/d/1Hs7VBjQC9OWn1a\_vEyaTExf97uORKBId/view. Unpublished manuscript} .
\bibitem[{Gonzalez-Eiras and Niepelt(2020)}]{gonzalez2020optimal}
\bibinfo{author}{Gonzalez-Eiras, M.}, \bibinfo{author}{Niepelt, D.},
  \bibinfo{year}{2020}.
\newblock \bibinfo{title}{On the optimal" lockdown" during an epidemic}.
\newblock \bibinfo{type}{Technical Report}. CESifo Working Paper.
\bibitem[{Gourinchas(2020)}]{gourinchas2020flattening}
\bibinfo{author}{Gourinchas, P.O.}, \bibinfo{year}{2020}.
\newblock \bibinfo{title}{Flattening the pandemic and recession curves}.
\newblock \bibinfo{journal}{Mitigating the COVID Economic Crisis: Act Fast and
  Do Whatever} \bibinfo{volume}{31}.
\bibitem[{Guerrieri et~al.(2020)Guerrieri, Lorenzoni, Straub and
  Werning}]{guerrieri2020macroeconomic}
\bibinfo{author}{Guerrieri, V.}, \bibinfo{author}{Lorenzoni, G.},
  \bibinfo{author}{Straub, L.}, \bibinfo{author}{Werning, I.},
  \bibinfo{year}{2020}.
\newblock \bibinfo{title}{Macroeconomic Implications of COVID-19: Can Negative
  Supply Shocks Cause Demand Shortages?}
\newblock \bibinfo{type}{Technical Report}. National Bureau of Economic
  Research.
\bibitem[{Harko et~al.(2014)Harko, Lobo and Mak}]{harko2014exact}
\bibinfo{author}{Harko, T.}, \bibinfo{author}{Lobo, F.S.},
  \bibinfo{author}{Mak, M.}, \bibinfo{year}{2014}.
\newblock \bibinfo{title}{Exact analytical solutions of the
  susceptible-infected-recovered (sir) epidemic model and of the sir model with
  equal death and birth rates}.
\newblock \bibinfo{journal}{Applied Mathematics and Computation}
  \bibinfo{volume}{236}, \bibinfo{pages}{184--194}.
\bibitem[{Hethcote(1989)}]{hethcote1989three}
\bibinfo{author}{Hethcote, H.W.}, \bibinfo{year}{1989}.
\newblock \bibinfo{title}{Three basic epidemiological models}, in:
  \bibinfo{booktitle}{Applied mathematical ecology}.
  \bibinfo{publisher}{Springer}, pp. \bibinfo{pages}{119--144}.
\bibitem[{Jenny et~al.(2020)Jenny, Jenny, Gorji, Arnoldini and
  Hardt}]{jenny2020dynamic}
\bibinfo{author}{Jenny, P.}, \bibinfo{author}{Jenny, D.F.},
  \bibinfo{author}{Gorji, H.}, \bibinfo{author}{Arnoldini, M.},
  \bibinfo{author}{Hardt, W.D.}, \bibinfo{year}{2020}.
\newblock \bibinfo{title}{Dynamic modeling to identify mitigation strategies
  for covid-19 pandemic}.
\newblock \bibinfo{journal}{medRxiv} .
\bibitem[{Jones et~al.(2020)Jones, Philippon and
  Venkateswaran}]{jones2020optimal}
\bibinfo{author}{Jones, C.J.}, \bibinfo{author}{Philippon, T.},
  \bibinfo{author}{Venkateswaran, V.}, \bibinfo{year}{2020}.
\newblock \bibinfo{title}{Optimal mitigation policies in a pandemic: Social
  distancing and working from home}.
\newblock \bibinfo{type}{Technical Report}. National Bureau of Economic
  Research.
\bibitem[{Kermack and McKendrick(1927)}]{kermack1927contribution}
\bibinfo{author}{Kermack, W.O.}, \bibinfo{author}{McKendrick, A.G.},
  \bibinfo{year}{1927}.
\newblock \bibinfo{title}{A contribution to the mathematical theory of
  epidemics}.
\newblock \bibinfo{journal}{Proceedings of the royal society of london. Series
  A, Containing papers of a mathematical and physical character}
  \bibinfo{volume}{115}, \bibinfo{pages}{700--721}.
\bibitem[{Kucharski et~al.(2020)Kucharski, Russell, Diamond, Liu, Edmunds,
  Funk, Eggo, Sun, Jit, Munday et~al.}]{kucharski2020early}
\bibinfo{author}{Kucharski, A.J.}, \bibinfo{author}{Russell, T.W.},
  \bibinfo{author}{Diamond, C.}, \bibinfo{author}{Liu, Y.},
  \bibinfo{author}{Edmunds, J.}, \bibinfo{author}{Funk, S.},
  \bibinfo{author}{Eggo, R.M.}, \bibinfo{author}{Sun, F.},
  \bibinfo{author}{Jit, M.}, \bibinfo{author}{Munday, J.D.}, et~al.,
  \bibinfo{year}{2020}.
\newblock \bibinfo{title}{Early dynamics of transmission and control of
  covid-19: a mathematical modelling study}.
\newblock \bibinfo{journal}{The lancet infectious diseases} .
\bibitem[{Liu et~al.(2020)Liu, Gayle, Wilder-Smith and
  Rockl{\"o}v}]{liu2020reproductive}
\bibinfo{author}{Liu, Y.}, \bibinfo{author}{Gayle, A.A.},
  \bibinfo{author}{Wilder-Smith, A.}, \bibinfo{author}{Rockl{\"o}v, J.},
  \bibinfo{year}{2020}.
\newblock \bibinfo{title}{The reproductive number of covid-19 is higher
  compared to sars coronavirus}.
\newblock \bibinfo{journal}{Journal of travel medicine} .
\bibitem[{Piguillem and Shi(2020)}]{piguillem2020optimal}
\bibinfo{author}{Piguillem, F.}, \bibinfo{author}{Shi, L.},
  \bibinfo{year}{2020}.
\newblock \bibinfo{title}{Optimal covid-19 quarantine and testing policies} .
\bibitem[{Tian et~al.(2020)Tian, Liu, Li, Wu, Chen, Kraemer, Li, Cai, Xu, Yang
  et~al.}]{tian2020investigation}
\bibinfo{author}{Tian, H.}, \bibinfo{author}{Liu, Y.}, \bibinfo{author}{Li,
  Y.}, \bibinfo{author}{Wu, C.H.}, \bibinfo{author}{Chen, B.},
  \bibinfo{author}{Kraemer, M.U.}, \bibinfo{author}{Li, B.},
  \bibinfo{author}{Cai, J.}, \bibinfo{author}{Xu, B.}, \bibinfo{author}{Yang,
  Q.}, et~al., \bibinfo{year}{2020}.
\newblock \bibinfo{title}{An investigation of transmission control measures
  during the first 50 days of the covid-19 epidemic in china}.
\newblock \bibinfo{journal}{Science} \bibinfo{volume}{368},
  \bibinfo{pages}{638--642}.
\bibitem[{Wang et~al.(2020)Wang, Fu, Zhang and Shi}]{wang2020evaluation}
\bibinfo{author}{Wang, N.}, \bibinfo{author}{Fu, Y.}, \bibinfo{author}{Zhang,
  H.}, \bibinfo{author}{Shi, H.}, \bibinfo{year}{2020}.
\newblock \bibinfo{title}{An evaluation of mathematical models for the outbreak
  of covid-19}.
\newblock \bibinfo{journal}{Precision Clinical Medicine} .
\bibitem[{Wu et~al.(2020)Wu, Leung and Leung}]{wu2020nowcasting}
\bibinfo{author}{Wu, J.T.}, \bibinfo{author}{Leung, K.},
  \bibinfo{author}{Leung, G.M.}, \bibinfo{year}{2020}.
\newblock \bibinfo{title}{Nowcasting and forecasting the potential domestic and
  international spread of the 2019-ncov outbreak originating in wuhan, china: a
  modelling study}.
\newblock \bibinfo{journal}{The Lancet} \bibinfo{volume}{395},
  \bibinfo{pages}{689--697}.

\end{thebibliography}

\end{document}